\documentclass[11pt,envcountsame,oribibl]{llncs}

\long\def\commentWeizmann #1\commentWeizmannend{}

\usepackage[T1]{fontenc} 

\usepackage{amssymb}
\usepackage{amsmath}
\usepackage[dvips]{epsfig}
\usepackage{algorithmic}

\newcommand{\cA}{{\cal A}}

\newcommand{\hide}[1]{}   

\begin{document}

\title{{More efficient periodic traversal in anonymous undirected graphs}}
\author{
J.~Czyzowicz\thanks{\tiny D\'epartement d'Informatique,
                          Universit\'e du Qu\'ebec en Outaouais,
                          Gatineau, Qu\'ebec J8X 3X7, Canada.
                          E-mail: {\tt jurek@uqo.ca}.}
\and
S.~Dobrev\thanks{\tiny Institute of Mathematics,
                       Slovak Academy of Sciences,
                       Dubravska 9, P.O.Box 56, 840 00,
                       Bratislava, Slovak Republic.
                       E-mail: {\tt stefan@ifi.savba.sk}.}
\and
L.~G\k{a}sieniec\thanks{\tiny Department of Computer Science,
                              University of Liverpool,
                              Ashton Street, Liverpool, L69 3BX, United Kingdom.
                              E-mail: {\tt \{L.A.Gasieniec,Russell.Martin\}@liverpool.ac.uk}.
                              L.~G\k{a}sieniec partially funded by the
                              Royal Society International Joint Project, IJP - 2007/R1.
                              R.~Martin partially funded by the Nuffield Foundation
                              grant NAL/32566, ``The structure and efficient
                              utilization of the Internet and other distributed systems''.}
\and
D.~Ilcinkas\thanks{\tiny CNRS, LABRI, Universit\'e Bordeaux I,
                         33 400 Talence, France.
                         E-mail: {\tt \{david.ilcinkas,ralf.klasing\}@labri.fr.}
                         Supported in part by the ANR projects ALADDIN and ALPAGE,
                         the INRIA project CEPAGE,
                         and the European projects GRAAL and DYNAMO.}
\and
J.~Jansson\thanks{\tiny Ochanomizu University, 2-1-1 Otsuka,
                        Bunkyo-ku, Tokyo 112-8610, Japan.
                        E-mail: \texttt{Jesper.Jansson@ocha.ac.jp}.
                        Funded by the Special Coordination Funds for
                        Promoting Science and Technology.}
\and
R.~Klasing$^{\dag}$
\and
I.~Lignos\thanks{\tiny Department of Computer Science,
                       Durham University, South Road, Durham, DH1 3LE, UK.
                       E-mail: {\tt yannis.lignos@durham.ac.uk}.}
\and
R.~Martin$^{\star\ \! \star\ \! \star}$
\and
K.~Sadakane\thanks{\tiny Principles of Informatics Research Division,
                         National Institute of Informatics.
                         2-1-2 Hitotsubashi, Chiyoda-ku, Tokyo 101-8430, Japan.
                         E-mail: {\tt sada@nii.ac.jp}.}
\and
W.-K.~Sung\thanks{\tiny Department of Computer Science,
                        National University of Singapore,
                        3 Science Drive 2, 117543 Singapore.
                        E-mail: {\tt ksung@comp.nus.edu.sg}.}
}

\institute{}

\date{}
\maketitle
\def\thefootnote{\fnsymbol{footnote}}

\begin{abstract}
  We consider the problem of {\em periodic graph exploration} in which
  a mobile entity with constant memory, {\sl an agent}, has to visit
  all $n$ nodes of an arbitrary undirected graph $G$ in a periodic
  manner. Graphs are supposed to be anonymous, that is, nodes are
  unlabeled. However, while visiting a node, the robot has to
  distinguish between edges incident to it. For each node $v$ the
  endpoints of the edges incident to $v$ are uniquely identified by
  different integer labels called {\sl port numbers}. We are interested in
  minimisation of the length of the exploration period.

  This problem is unsolvable if the local port numbers are set
  arbitrarily, see~\cite{Bu78}. However, surprisingly small periods
  can be achieved when assigning carefully the local port
  numbers. Dobrev et al.~\cite{DJSS05} described an algorithm for
  assigning port numbers, and an oblivious agent (i.e. agent with no
  memory) using it, such that the agent explores all graphs of size
  $n$ within period $10n$. Providing the agent with a constant number
  of memory bits, the optimal length of the period was proved
  in~\cite{GKMNZ08} to be no more than $3.75n$ (using a different
  assignment of the port numbers). In this paper, we improve both
  these bounds. More precisely, we show a period of length at most
  $4\frac{1}{3}n$ for oblivious agents, and a period of length at most
  $3.5n$ for agents with constant memory. Moreover, we give the first
  non-trivial lower bound, $2.8n$, on the period length for the
  oblivious case.
\end{abstract}




\section{Introduction}

Efficient search in unknown or unmapped environments is one of the
fundamental problems in algorithmics. Its applications range from
robot navigation in hazardous environments to rigorous exploration
(and, e.g., indexing) of data available on the Internet. Due to a
strong need to design simple and cost effective agents as well as to
design exploration algorithms that are suitable for rigorous
mathematical analysis, it is of practical importance to limit the
local memory of agents.

We consider the task of graph exploration by a mobile entity
equipped with small (constant number of bits) memory. The mobile
entity may be, e.g., an autonomous piece of software navigating
through an underlying graph of connections of a computer network.
The mobile entity is expected to visit all nodes in the graph in a
periodic manner. For the sake of simplicity, we call the mobile
entity an {\em agent} and model it as a finite state automaton. The
task of periodic traversal of all nodes of a network is particularly
useful in network maintenance, where the status of every node has to
be checked regularly.

We consider here undirected graphs that are anonymous, i.e., the nodes
in the graph are neither labelled nor colored. To enable the agent to
distinguish the different edges incident to a node, edges at a node
$v$ are assigned {\em port numbers} in $\{1,\dots,d_v\}$ in a
one-to-one manner, where $d_v$ is the degree of node $v$.

We model agents as {\em Mealy
  automata}. The Mealy automaton has a finite number of states and a
transition function $f$ governing the actions of the agent.
If the automaton enters a node $v$ of degree $d_v$ through
port $i$ in state $s$, it switches to state $s'$ and exits the node
through port $i'$, where $(s',i')=f(s,i,d_v)$. The memory size of an
agent is related to its number of states, more precisely it equals the
number of bits needed to encode these states. For example an oblivious
agent has a single state, or equivalently zero memory bits. Note that
the size of the agent memory represents in this model the amount of
information that the agent can remember while moving. This does not
restrict computations made on a node and thus the transition function
can be any deterministic function. Additional memory needed for
computations can be seen as provided temporarily by the hosting
node. Nevertheless, our agent algorithms perform very simple tests and
operations on the non-constant inputs $i$ and $d$, namely equality
tests and incrementations.

Periodic graph exploration requires that the agent has to visit
every node infinitely many times in a periodic manner. In this
paper, we are interested in minimising the length of the
exploration period. In other words, we want to minimise the maximum
number of edge traversals performed by the agent between two
consecutive visits of a generic node, while the agent enters this
node in the same state through the same port.

However, Rollik~\cite{Rol80} proved that this problem is
unsolvable as an agent needs $\Omega(\log n)$ memory bits to explore
all graphs of size $n$, even restricted to cubic planar graphs. This
lower bound has been proved recently to be actually optimal by
Reingold in his breakthrough paper~\cite{Rein}. Providing the agent
with a pebble to mark nodes does not help much as the asymptotic size
of memory needed remains $\Omega(\log n)$
bits~\cite{FraIlcRajTix06}. In fact, even a highly-coordinated
multi-agent team capable of (restricted) teleportation cannot explore
all graphs with constant memory~\cite{CooRac80}.

Nevertheless, putting some information in the graph does help a
lot. Cohen et al.~\cite{CohFraIlcKorPel05} showed that putting two
bits of advice at each node allows to explore all graphs by an agent
with constant memory, by a periodic traversal of length $O(m)$, where $m$ is the
number of edges. In fact, the impossibility results presented above
all use the ability of the adversary to assign the local port numbers
in a misleading order. On the other hand, even if nodes are not marked
in any way but if port numbers are carefully assigned (still
satisfying the condition that at each node $v$, port numbers from $1$
to $d_v$ are used), then a simple agent, even oblivious, can perform
periodic graph exploration within period of length $O(n)$. Using
appropriate assignment of the local port numbers, the best known
period achieved by an oblivious agent is $10n$~\cite{DJSS05} whereas
the best known period achieved by an agent with constant memory is
$3.75n$~\cite{GKMNZ08}.

\subsection{Our results}

In this paper, we improve both
these bounds. More precisely, we present an efficient deterministic
algorithm assigning port numbers in the graph, such that, an oblivious
agent is able to accomplish each period of the traversal route in at
most $4\frac{1}{3}n$.
Our algorithm uses a new three-layer partition of graphs permitting an
optimal $O(|E|)-$time construction of the port labeling.
As a complement, we present a class of graphs in which an oblivious agent
performs a tour of at most $2n$.
In addition, we present another
algorithm assigning port numbers in the graph,
also using the three-layer partitioning approach,
such that, an agent with
constant memory is able to accomplish periodic graph exploration within period
at most $3.5n$.
Moreover, we give the first
non-trivial lower bound, $2.8n$, on the period length for the
oblivious case.

\subsection{Related Work}\label{Sec-Rel}
Graph exploration by robots has recently attracted growing
attention. The unknown environment in which the robots operate is
often modelled as a graph, assuming that the robots may only move
along its edges. The graph setting is available in two different
forms.

In~\cite{AlbHe00,BenFRSV98,BenSl94,DePa99,FT}, the robot explores
strongly connected directed graphs and it can move only in one
pre-specified direction along each edge.
In~\cite{AwBS99,BRS2,CohFraIlcKorPel05,DunKK01,FraIPPP04,FraIlcRajTix06,PanPe99}, the
explored graph is undirected and the agent can traverse edges in
both directions.
Also, two alternative efficiency measures are adopted in most papers
devoted to graph exploration, namely, the {\em time} of completing
the task~\cite{AlbHe00,AwBS99,BenFRSV98,BenSl94,BRS2,DePa99,DunKK01},
or the number of {\em
memory bits} (states in the automaton) available to the agent
\cite{CohFraIlcKorPel05,DiFKP02,FraIl04,FraIPPP04,FraIlcRajTix06,GPRZ07}.

Graph exploration scenarios considered in the literature differ in
an important way: it is either assumed that nodes of the graph have
unique labels which the agent can recognise, or it is assumed that
nodes are anonymous. Exploration of directed graphs assuming the
existence of labels was investigated in~\cite{AlbHe00,DePa99,FT}. In this
case, no restrictions on the agent moves were imposed, other than by
directions of edges, and fast exploration algorithms were sought.
Exploration of undirected labelled graphs was considered
in~\cite{AwBS99,AK,BRS2,DunKK01,PaPe}. Since in this case a simple
exploration based on depth-first search can be completed in time
$2m$, where $m$ is the number of edges, investigations concentrated
either on further reducing the time for an unrestricted agent, or on
studying efficient exploration when moves of the agent are
restricted in some way.  The first approach was adopted in
\cite{PaPe}, where an exploration algorithm working in time
$m+O(n)$, with $n$ being the number of nodes, was proposed.
Restricted agents were investigated in~\cite{AwBS99,AK,BRS2,DunKK01}. It
was assumed that the agent is a robot with either a restricted fuel
tank~\cite{AwBS99,BRS2}, forcing it to periodically return to the base
for refuelling, or that it is a tethered robot, i.e., attached to
the base by a ``rope'' or ``cable'' (a path from the original node)
of restricted length~\cite{DunKK01}. For example, in~\cite{DunKK01} it was
proved  that exploration can be done in time $O(m)$ under both
scenarios.

Exploration of anonymous graphs presents different types of challenges.
In this case, it is impossible to explore arbitrary graphs
and to stop after completing exploration if no marking of nodes is
allowed~\cite{Bu78}. Hence, the scenario adopted in~\cite{BenFRSV98,BenSl94}
was to allow {\em pebbles} which the agent can drop on nodes to
recognise already visited ones, and then remove them and drop in
other places. The authors concentrated attention on the minimum
number of pebbles allowing efficient exploration of arbitrary
directed $n$-node graphs. (In the case of undirected graphs, one
pebble suffices for efficient exploration.) In~\cite{BenSl94}, the
authors compared the exploration power of one agent with pebbles to
that of two cooperating agents without pebbles. In~\cite{BenFRSV98}, it
was shown that one pebble is enough, if the agent knows an upper
bound on the size of the graph, and $\Theta(\log \log n)$ pebbles
are necessary and sufficient otherwise.

In~\cite{CohFraIlcKorPel05,DiFKP02,FraIl04,FraIPPP04,FraIlcRajTix06},
the adopted measure of
efficiency was the memory size of the agent exploring anonymous
graphs. In~\cite{FraIl04,FraIlcRajTix06}, the agent was allowed to mark nodes by
pebbles, or even by writing messages on whiteboards with which nodes
are equipped. In~\cite{CohFraIlcKorPel05}, the authors studied special schemes
of labelling nodes, which facilitate exploration with small memory.
Another aspect of distributed graph exploration by robots with
bounded memory was studied in~\cite{DiFKP02,GPRZ07}, where the topology
of graphs is restricted to trees. In~\cite{DiFKP02} Diks {\em et al.}
proposed a robot requiring $O(\log^2 n)$ memory bits to explore any
tree with at most $n$ nodes. They also provided the lower bound
$\Omega(\log n)$ if the robot is expected to return to its original
position in the tree. Very recently the gap between the upper bound
and the lower bound was closed in~\cite{GPRZ07} by G\k asieniec {\em
et al.} who showed that $O(\log n)$ bits of memory suffice in tree
exploration with return. However it is known, see~\cite{FraIPPP04}, that in
arbitrary graphs the number of memory bits required by any robot
expected to return to the original position is $\Theta(D\log d),$
where $D$ is the diameter and $d$ is the maximum degree in the
graph. In comparison, Reingold~\cite{Rein} proved recently that $SL =
L$, i.e., any decision problem which can be solved by a
deterministic Turing machine using logarithmic memory (space) is
log-space reducible to the USTCON (st-connectivity in undirected
graphs) problem. This proves the existence of a robot equipped with
asymptotically optimal number of $O(\log n)$ bits being able to
explore any $n$-node graph in the perpetual exploration model, where
the return to the original position is not required. The respective
lower bound $\Omega(\log n)$ is provided in~\cite{Rol80}.

In this paper, we are interested in robots characterised by very low
memory utilisation. In fact, the robots are allowed to use only a
constant number of memory bits. This restriction permits modelling
robots as finite state automata. Budach~\cite{Bu78} proved that no
finite automaton can explore all graphs. Rollik~\cite{Rol80} showed
later that even a finite team of finite automata cannot explore all
planar cubic graphs. This result is improved in~\cite{CooRac80}, where
Cook and Rackoff introduce a powerful tool, called the {\em JAG},
for Jumping Automaton for Graphs. A JAG is a finite team of finite
automata that permanently cooperate and that can use {\em
teleportation} to move from their current location to the location
of any other automaton. However, even JAGs cannot explore all
graphs~\cite{CooRac80}.

\section{Preliminaries}

\subsection{Notation and basic definitions}

Let $G=(V,E)$ be a simple, connected, undirected graph.
We denote by $\overrightarrow{G}$ the symmetric directed graph obtained
from $G$ by replacing each undirected edge $\{u,v\}$ by two directed edges
in opposite directions -- the directed edge from $u$ to $v$ denoted by $(u,v)$
and the directed edge from $v$ to $u$ denoted by $(v,u)$. For each directed
edge  $(u,v)$ or $(v,u)$ we say that undirected edge $\{u,v\}\in G$ is
its {\em underlying} edge.
For any node $v$ of a directed graph the {\em out-degree} of $v$ is the number
of directed edges leaving $v$, the {\em in-degree} of $v$ is
the number of directed edges incoming to $v$, and {\em cumulative degree}
of $v$ is the sum of its out-degree and its in-degree.

Directed cycles constructed by our algorithm traverse some edges in $G$
once and some other edges twice in opposite directions.
However, at early stages, our algorithm for oblivious agents is solely
interested whether the edge is unidirectional or bidirectional,
indifferently of the direction. To alleviate the presentation (despite some
abuse of notation), in this context, an edge that is traversed once when
deprived of its direction we call a {\sl single edge}.
Similarly, an edge that is traversed twice is called a {\sl two-way edge},
and it is understood to be composed of two single edges (in opposite directions).
Hence we extend the notion of single and two-way edges to general directed graphs
in which the direction of edges is removed. In particular,
%
%
we say that two remote nodes $s$ and $t$ are connected by
a {\sl two-way path}, if there is
a finite sequence of vertices $v_1,v_2,\dots,v_k,$
where each pair $v_i$ and $v_{i+1}$ is connected by a two-way edge, and
$s=v_1$ and $t=v_k.$
We call a directed graph $\overrightarrow{K}$ {\em two-way connected} if for
any pair of nodes there is a two-way path connecting them.
Note that two-way connectivity implies strong connectivity but not the opposite.

\subsection{Three-layer partition\label{s:3l-partition}}


%
The three-layer partition is a new graph decomposition method that we
use in to efficiently construct periodic tours in both the oblivious and
the bounded-memory cases.

For any set of nodes $X$ we call the {\em neighborhood} of $X$ the set
of their neighbors in graph $G$ (excluding nodes in $X$) and we
denote it by $N_G(X)$.
One of the main components of the constructions of our technique
are {\em backbone trees} of $G$, that is connected cycle-free subgraphs of $G$.
We say that a node $v$ is {\em saturated} in a backbone tree $T$ of $G$ if all
edges incident to $v$ in $G$ are also present in $T$.

A {\em three-layer partition} of a graph $G=(V,E)$ is a
4-tuple $(X,Y,Z,T_B)$ such that
(1) the three sets $X$, $Y$ and $Z$ form a partition of $V$,
(2) $Y=N_G(X)$ and $Z=N_G(Y)\setminus X$,
(3) $T_B$ is a tree of node-set $X\cup Y$ where all nodes in $X$ are saturated.
We call $X$ the {\sl top layer}, $Y$ the {\sl middle layer},
and $Z$ the {\sl bottom layer} of the partition. Any edge of $G$
between two nodes in $Y$ will be called {\sl horizontal}.

During execution of procedure {\sc 3L-Partition} the nodes in $V$ are dynamically
partitioned into sets $X,Y,Z,P$ and $R$ with temporary contents, where $X$ is the set of
saturated nodes, $Y=N_G(X)$ contains nodes at distance 1 from $X$,
$Z=N_G(Y)\setminus X$ contains
nodes at distance 2 from $X$, $P=N_G(Z)\setminus Y$ contains nodes
at distance 3
from $X$ and $R=V\setminus (X\cup Y\cup Z\cup P)$ contains all the remaining nodes in~$V$.
%
%
\vspace*{-0.2cm}
\begin{tabbing}
xxx\=xxx\=xxx\=xxx\=xxx\=xxx\kill\\
{\bf Procedure} {\sc 3L-Partition}$(in: G=(V,E); out: X,Y,Z,T_{B});$\\
(1) \> $X=Y=Z=P=\emptyset;$ $R=V;$ $T_{B}=\emptyset;$\\
(2) \> select an arbitrary node $v\in R;$\\
(3) \> {\bf loop}\\
\> (a) $X=X\cup\{v\};$ (insert into $X$ newly selected node);\\
\> (b) update contents of sets $Y,Z,P$ and $R$ (on the basis of new $X$);\\
\> (c) saturate the newly inserted node $v$ to $X$ (i.e., insert
           all new edges to $T_{B}$);\\
\> (d) {\bf if} the new node $v$ in $X$ was selected from $P$ {\bf then}
           insert to $T_{B}$ an arbitrary horizontal edge\\
\>\>\> (on middle level) to connect the newly formed star rooted
           in $v$ with the rest of $T_{B}.$\\
\> (e) {\bf if} any new node $v\in Y$ can be saturated {\bf then}
           select $v$ for saturation;\\
\>\>\> {\bf else-if} any new node $v\in Z$ can be saturated {\bf then}
           select $v$ for saturation;\\
\>\>\>\> {\bf else-if} $P$ is non-empty {\bf then} select a new $v$
           from $P$ for saturation arbitrarily;\\
\>\>\>\>\> {\bf else} exit-loop;\\
\> {\bf end-loop}\\
(4) \> {\bf output} $(X,Y,Z,T_{B})$
\end{tabbing}

\vspace*{-0.3cm}
\begin{figure}[htbp] 
   \centering
   \includegraphics[width=4in]{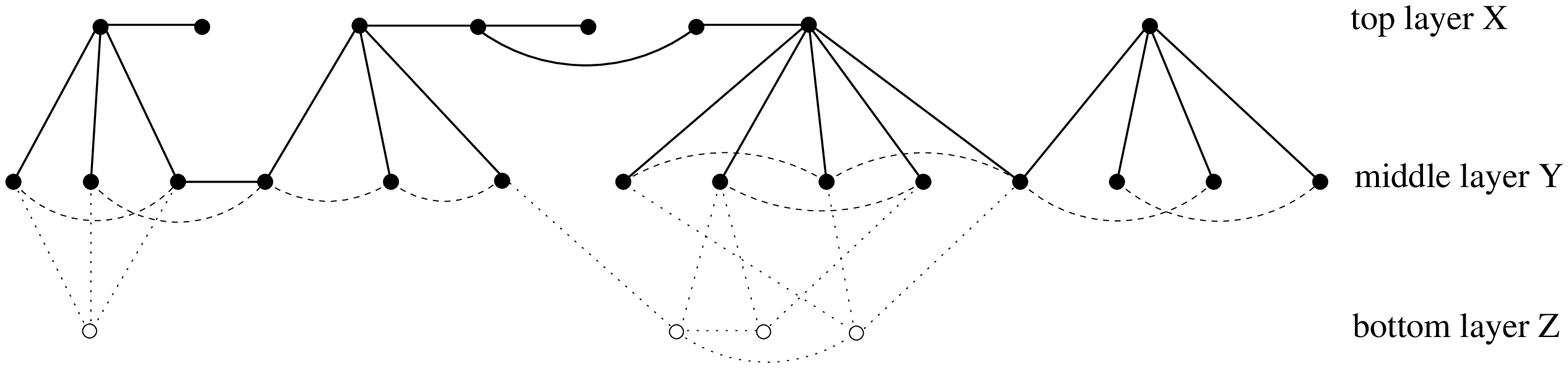}
   \caption{Three-layer partition. Solid lines and black nodes belong
   to the backbone tree $T_B$.
   Dashed lines represent horizontal edges outside $T_B$. Dotted lines
   are incident to nodes from $Z.$}
   \label{fig:3l-partition}
\end{figure}

\begin{lemma}\label{lem:algo}
Procedure {\sc 3L-Partition} computes a three-layer
partition for any connected graph $G.$
\end{lemma}

\begin{lemma}\label{lm:3-conditions}
The three-layer partition has the following properties:
\begin{description}
%
%
\vspace{-0.2cm}
\item[(1)] each node in $Y$ has an incident {\sl horizontal} edge outside of $T_{B}$;
\vspace{-0.2cm}
\item[(2)] each node in $Z$ has at least two neighbors in $Y$.
\end{description}
\end{lemma}
\begin{proof}
%
To prove property (1) assume, by contradiction, that there exists
a node $u\in Y$ that has no horizontal edges outside of $T_B.$ Observe that
in this case $u$ can be saturated , i.e., $u$ may be moved to $X,$
inserting into $T_B$ all remaining
edges incident to $u.$ Indeed, since before $u$ was saturated all such edges lead only
to nodes in $Z$ their insertion does not form cycles. Thus property (1) holds.
Finally, assume there is a node $w$ in $Z$ with no more than one incident edge leading
to level $Y.$ Also in this case we can saturate $w$ since all edges incident to $w$
form a star that shares at most one node with $T_B.$ Thus, no cycle is
created, which in turn proves property (2). \qed
\end{proof}

%



\begin{lemma}\label{lm:3layer-complexity}
For any graph $G=(V,E)$ a three-layer partition may be computed in $O(|E|)$ time.
\end{lemma}

\subsection{RH-traversability and witness cycles}
In this section we discuss the conditions for the oblivious periodic
traversals. Given a port number assignment algorithm and an agent
algorithm, it is possible, for a given degree $d$, to permute all port
numbers incident to each degree-$d$ node of a graph $G$ according to
some fixed permutation $\sigma$, and to modify the transition function
$f$ of the agent accordingly, so that the agent behaves exactly the
same as before in $G$. The new transition function $f'$ is in this
case given by the formula $f'=\sigma \circ f \circ \sigma^{-1}$ and
the two agent algorithms are said to be equivalent.

More precisely, two agent algorithms described by their respective
transition functions $f$ and $f'$ are {\em equivalent} if for any
$d>0$ there exists a permutation $\sigma$ on $\{1,\dots,d\}$ such that
$f'=\sigma \circ f \circ \sigma^{-1}$.

The most common algorithm used for oblivious agents is the
Right-Hand-on-the-Wall algorithm. This algorithm is specified by the
transition function $f:(s,i,d)\mapsto (s,(i \bmod d) +1)$. Differently
speaking, if the agent enters a degree-$d$ node $v$ by port number
$i$, it will exit $v$ through port number $(i \bmod d) +1$.

The following lemma states that any couple consisting of a port number
assignment algorithm and an oblivious agent algorithm, and solving the
periodic graph exploration problem, can be expressed by using the
Right-Hand-on-the-Wall algorithm as the agent algorithm. We will thus
focus on this algorithm in all subsequent parts referring to oblivious
agents.

\begin{lemma}\label{lem:only-RH}
  Any agent algorithm enabling an oblivious agent to explore all
  graphs (even all stars) is equivalent to the Right-Hand-on-the-Wall
  algorithm.
\end{lemma}


Graph traversal according to the Right-Hand-on-the-Wall
  algorithm has been called {\em right-hand traversals} or shortly {\em RH-traversals},
see~\cite{DJSS05}. Similarly, cyclic paths formed in the graph according
to the right-hand rule are called {\em RH-cycles}.
The aim of our first oblivious-case algorithm is to find a short RH-traversal of the graph, i.e., to find a cycle $\overrightarrow{C}$ in
$\overrightarrow{G}$ containing all nodes of $\overrightarrow{G}$ and satisfying the
right-hand rule: If $e_1=(u,v)$ and $e_2=(v,w)$ are two
successive edges of $\overrightarrow{C}$ then $e_2$ is the successor
of $e_1$ in the port numbering of $v$. We call such a
cycle a {\em witness cycle} for $G$, and the corresponding port
numbering a {\em witness port numbering}.

Given graph $\overrightarrow{G}$ we first design $\overrightarrow{H}$, a
spanning subgraph of $\overrightarrow{G}$ that contains all edges of
a short witness cycle $\overrightarrow{C}$ of $\overrightarrow{G}$.
Then we look for port numbering of each node in
$\overrightarrow{H}$ to obtain $\overrightarrow{C}$. The
characterisation of such a graph $\overrightarrow{H}$ is not trivial,
however it is easy to characterise
graphs which are unions of RH-cycles.

\begin{definition} A node $v \in \overrightarrow{G}$ is {\em RH-traversable}
in $\overrightarrow{H}$ if there exists a port numbering $\pi_v$
such that, for
each edge $(u,v) \in \overrightarrow{H}$ incoming to $v$ via an underlying
edge $e$ there exists an outgoing edge $(v,w) \in \overrightarrow{H}$ leaving~$v$ via
the underlying edge $e'$, such that $e'$ is the successor of $e$ in
the port numbering of~$v$.

We call such ordering a {\em witness ordering} for $v$.
\end{definition}

Let $\overrightarrow{H}$ be a spanning subgraph of $\overrightarrow{G}$. For each node $v$,
denote by $b_v$, $i_v$ and $o_v$ the number of two-way edges incident
to $v$ used in $\overrightarrow{H},$ only incoming and only outgoing edges,
respectively.
The following lemma characterises the nodes of a graph being an union of RH-cycles.

\begin{lemma} \label{lm:RHchar}
A node $v$ is RH-traversable if and only if $b_v=d_v$ or $i_v=o_v>0$.
\end{lemma}

\begin{proof}
$(\Rightarrow)$ The definition of RH-traversability implies $i_v=o_v$.

\noindent
$(\Leftarrow)$ If $b_v=d_v,$ i.e., all edges incident to
$v$ are used in two directions, any ordering of the edges is acceptable.
Otherwise ($b_v\ne d_v,$), choose a port numbering in which outgoing edges that contribute
to two-way edges are arranged in one block followed by an outgoing edge.
All remaining directed edges are placed in a separate block, in which
edges alternate directions and the last (incoming) edge precedes the block
of all two-way edges. \qed
\end{proof}

\vspace*{-0.15cm}
We easily obtain the following

\begin{corollary} \label{cor:RHchar}
A spanning subgraph $\overrightarrow{H}$ of $\overrightarrow{G}$ is a union of RH-cycles
if and only if each node $v$ has an even number of single edges
incident to $v$ in $\overrightarrow{H}$, and, in case no single
edge is incident to $v$ in $\overrightarrow{H}$, all two-way edges
incident to $v$ in $\overrightarrow{G}$ must be also present in $\overrightarrow{H}$.
\end{corollary}

In the rest of this section we introduce several operations on cycles,
and the conditions under which these operations will result in a witness cycle.


Consider a subgraph $\overrightarrow{H}$ of $G$ that has only RH-traversable nodes.
Observe that any
port numbering implies a partitioning of $\overrightarrow{H}$ into a set of RH-cycles.
Take any ordering $\gamma$ of this set of cycles. We define two rules
which transform one set of cycles to another. The first rule, {\em Merge3},
takes as an input three cycles
incident to a node and merge them to form a single one. The
second rule, {\em EatSmall}, breaks a non-simple cycle into two
sub-cycles and transfers one of them to another cycle.


\begin{figure}
\begin{center}
\includegraphics[scale=0.76]{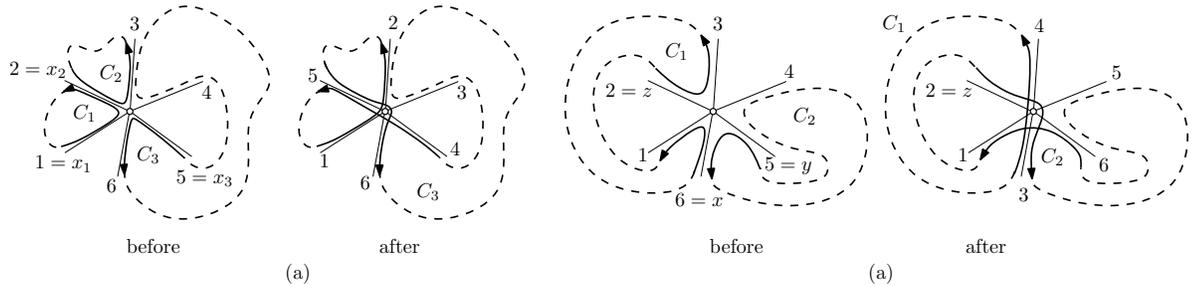}
\end{center}
\caption{(a) Applying rule {\em Merge3}; (b) applying rule {\em EatSmall}.\label{Fig:merge3}\label{Fig:eatSmall}}
\end{figure}

\begin{enumerate}

\item
{\bf Rule Merge3:} Let $v$ be a node incident to
at least three different cycles $C_1$, $C_2$ and $C_3$. Let $x_1$,
$x_2$ and $x_3$ be the underlying edges at $v$ containing incoming edges
for cycles $C_1$, $C_2$ and $C_3$, respectively ($x_1,x_2$ and $x_3$
can be a single edge or a two-way edge in $\overrightarrow{H}$). Suppose w.l.o.g.,
that $x_2$ is between $x_1$ and $x_3$ in cyclic port numbering of $v$. The port numbering
which makes the successor of $x_2$ becomes the successor of
$x_1$, the successor of $x_3$ becomes the successor of $x_2$ and the
successor of $x_1$ becomes the successor of $x_3$ and keeps the relative
order of the remaining edges the same (see Figure \ref{Fig:merge3}(a))
connects the cycles $C_1$, $C_2$ and $C_3$ into a single cycle $C_3$, while
remaining a witness port numbering for $v$ (due to the original port numbering).

\item
{\bf Rule EatSmall:} Let $C_1$ be the smallest cycle in ordering $\gamma$ such that
\begin{itemize}
\item there is a node $v$ that appears in $C_1$ at least twice
\item there is also another cycle $C_2$ incident to $v$
\item $\gamma(C_1)<\gamma(C_2)$
\end{itemize}
Let $x$ and $y$ be underlying edges at $v$
containing incoming edges for $C_1$ and $C_2,$ respectively;
let $z$ be the underlying edge containing the incoming edge by which
$C_1$ returns to $v$ after leaving via the successor of $x$. If $z$
is the successor of $y$, choose a different $x$. Modify the ordering of
the edges in $v$ as follows: (1) the successor of $x$ becomes the
new successor of $y$, (2) the old successor of $y$ becomes the new
successor of $z$, (3) the old successor of $z$ becomes the new
successor of $x$ and (4) the order of the other edges remains
unchanged -- see Figure \ref{Fig:eatSmall}(b).

\end{enumerate}


\begin{lemma}\label{lm:main}
Let $\overrightarrow{K}$ be a two-way connected spanning subgraph of  $G$
with all nodes RH-traversable in  $\overrightarrow{K}$. Consider the set
of RH-cycles generated by some witness port numbering of its nodes,
with $C^*$ being the largest cycle according to some ordering $\gamma$.
If neither {\em Merge3} nor {\em EatSmall} can be applied to
the nodes of $C^*$ then $C^*$ is a witness cycle.
\end{lemma}

\begin{proof}
Suppose, by contradiction, that $C^*$ does not span all the nodes in $G$.
Let $V'$ be the set of nodes of $G$ not traversed by $C^*$.
Since $\overrightarrow{K}$ is two-way connected there exist two nodes $u,v \in G$,
such that $v$ belongs to $C^*$ and $u \in V'$, and the directed edges $(u,v)$
and $(v, u)$ belong to $\overrightarrow{K}$. Edges $(u,v)$ and $(v,u)$ cannot belong
to different cycles of $\overrightarrow{K}$ because {\em Merge3} would be applicable.
Hence $(u,v)$  and $(v, u)$ must both belong to the same cycle $C'$. However $(u,v)$
and $(v, u)$ cannot be consecutive edges of $C'$ because this would imply $d_v=1$
which is not the case, since $v$ also belongs to $C^*$. Hence $C'$ must visit $v$
at least twice.
However, since $C^*$ is the largest cycle we have $\gamma(C') < \gamma(C^*)$
and the conditions of applicability of rule {\em EatSmall} are satisfied with
$C_1=C'$ and $C_2=C^*$. This is the contradiction proving the claim
of the lemma. \hfill \qed
\end{proof}



\section{Oblivious periodic traversal}

\noindent
In this section we describe the algorithm that constructs a short witness
cycle for graph $G$. According to lemma~\ref{lm:main} it is
sufficient to construct a spanning subgraph $\overrightarrow{K}$ of $G$
which is two-way connected, such that,
each node of $G$ is RH-traversable in $\overrightarrow{K}$. We will
present first a restricted case of a {\em terse set of RH-cycles}, when
it is possible to construct a spanning tree of $G$ with no saturated
node. In this case we can construct a witness cycle of size $2n$.
In the case of arbitrary graphs, we need a more involved argument,
which will lead to a witness cycle of size $4 \frac{1}{3}n$. We
conclude this section with the presentation of
a lower bound of $2.8n$.

\subsection{Terse set of RH-cycles}

\noindent
Suppose that we have a graph $G$, which has a spanning tree $T$ with no saturated node.
This happens for large and non-trivial classes of graphs, including two-connected graphs,
graphs admitting two disjoint spanning trees, and many others. For those graphs we present
an algorithm that finds a shorter witness cycle than one that we can find for arbitrary
graphs. The idea of the algorithm is to first construct a spanning subgraph of $G,$
$\overrightarrow{K}$ of size $2n$,
which contains only RH-traversable
nodes (cf.\ algorithm \textsc{TerseCycles}). Then we apply a port numbering
which partitions $\overrightarrow{K}$ into a set of RH-cycles that can then
be merged into a single witness cycle (cf. Corollary~\ref{cor:TerseCycles}).

\smallskip\smallskip\smallskip
\noindent
{\bf Algorithm} \textsc{TerseCycles}:
\begin{algorithmic}[1]
\STATE Find $T$ -- a spanning subgraph of $G$ with no saturated nodes;
\STATE $\overrightarrow{K} \leftarrow T;$ \{each edge in $T$ is a two-way edge
in $\overrightarrow{K}$\}
\STATE For each node $v \in \overrightarrow{K}$ add to $\overrightarrow{K}$ a
single edge from $G \setminus T;$ \{the single edges form a collection of stars $S$\}
\STATE {\sc Restore-Parity}$(\overrightarrow{K}, T, root(T));$
\end{algorithmic}

\vspace*{0.2cm}
Procedure {\sc Restore-Parity}  has to assure that the number of single edges incident
to each node is even. The procedure visits each node $v$ of the tree $T$ in the bottom-up
manner and counts all single edges incident to $v.$ If this number is odd, the two-way edge
leading to the parent is reduced to a single edge (with the direction to be specified
later). The procedure terminates when the parity of all children of the root
in the spanning tree is restored. Note also that the cumulative degree of the
root must be even since
the cumulative degree of all nodes in $S$ is even. Note also that no decision about
the direction of single edges is made yet.

\bigskip\noindent{\bf Procedure}
\textsc{RestoreParity}(directed graph $\overrightarrow{K},$ tree $T,$ node $v$): integer;
\begin{algorithmic}[1]
\STATE $P_v =$ (number of single edges in $\overrightarrow{K} \setminus T) \pmod 2$;
\IF {$v$ is not a leaf in $T$}
\FOR {each node $c_v \in T$ being a child of $v$}
\STATE $P_v \leftarrow (P_v + RestoreParity(\overrightarrow{K}, T, c_v)) \pmod 2$;
\ENDFOR
\ENDIF
\IF {$P_v = 1$}
\STATE  reduce the two-way edge $(P, parent(P))$ to single;
\ENDIF
\STATE {\bf return} $P_v$;
\end{algorithmic}

\begin{lemma}
After the completion of procedure \textsc{TerseCycles}
every node of $\overrightarrow{K}$ is RH-traversable.
\end{lemma}
\begin{proof}
Every node is either saturated or it has at least
two single edges
incident to it. \qed
\end{proof}

\begin{corollary}\label{cor:TerseCycles}
For any graph $G$ admitting a spanning tree $T$, such that none of the nodes
is saturated (i.e., $G \setminus T$ spans all nodes of $G$) it is
possible to construct a witness cycle of length at most $2n.$
\end{corollary}
%

Corollary~\ref{cor:TerseCycles} gives small witness cycles for a large class of graphs.
It should be noted for $3$-regular graphs, finding a spanning tree having no
saturated nodes corresponds to finding a Hamiltonian path, a problem known to be
NP-hard even in this restricted setting~\cite{Garey}.

\subsection{Construction of witness cycles in arbitrary graphs}
\noindent
The construction of witness
cycle is based on the following approach. First select a spanning tree $T$ of graph $G$
composed of two-way edges. Let $G_i$, for $i=1, 2, \dots, k$ be the connected components
of $G \setminus T$, having, respectively, $n_i$ nodes. For each such component we apply
procedure \textsc{3L-Partition}, obtaining three sets $X_{i},Y_{i}$ and $Z_{i}$ and
a backbone tree $T_{i}$. We then add single edges incident to the
nodes of sets $Y_{i}$ and $Z_{i}$, and we apply the procedure \textsc{RestoreParity}
to each component $G_i$. We do this in such a way that the total number of edges in
$G_i$ is smaller than $2 \frac{1}{3}n$. For the union of graphs
$T \cup G_1 \cup G_2 \cup \dots \cup G_k$ we take a port numbering
that generates a set of cycles.
The port numbering and orientation of edges in the union of graphs is obtained as follows.
First we remove temporarily all two-way edges from the union.
The remaining set of single edges is partitioned into a collection of simple cycles,
where edges in each cycle have a consistent orientation. Further we reinstate all two-way
edges in the union, s.t., each two-way edge is now represented as two
arcs with the opposite direction. Finally we provide port numbers at each
node of the union, s.t., it is consistent
with the RH-traversability condition, see lemma~\ref{lm:RHchar}.
We apply rules {\em Merge3} and {\em EatSmall} to this set of cycles until neither rule
can be applied. The set of cycles obtained will contain
a witness cycle, using lemma~\ref{lm:main}.

\bigskip\noindent{\bf Algorithm} \textsc{FindWitnessCycle};
\begin{algorithmic}[1]
\STATE Find a spanning tree $T$ of graph $G$ \{two-way edges\}
\FOR {each connected component $G_i$ of $G \setminus T$}
\STATE {\sc 3L-Partition}$(G_i,X_i,Y_i,Z_i,T_i)$;
\STATE Form set $P_i$ by selecting for each node in $Z_i$ two edges
           leading to $Y_i$; \{single edges\};
\STATE Form a set of independent stars $S_i$ spanning all nodes in $Y_i$
that are not incident to $P_i$; \{single edges\};
\STATE {\sc RestoreParity}$(G_i \cup P_i \cup S_i, T_i, root(T_i))$;
\ENDFOR
\STATE $\overrightarrow{K} \leftarrow T \cup G_1 \cup G_2 \cup \dots \cup G_k$; \label{witness-K}
\STATE Take any port numbering and produce a set
$\mathfrak{C}$ of RH-cycles induced by it;
\STATE  Apply repeatedly {\em Merge3} or, if not possible, {\em EatSmall} to $\mathfrak{C}$ until neither rule can be applied;
\STATE {\bf return} the witness cycle of $\mathfrak{C}$;
\end{algorithmic}

\begin{theorem}\label{th:4.3}
For any $n$-node graph algorithm \textsc{FindWitnessCycle} returns a
witness cycle of size at most $4\frac{1}{3}n-4$.
\end{theorem}

\begin{theorem} \label{th:complexity}
The algorithm \textsc{FindWitnessCycle} terminates in $O(|E|)$ time.
\end{theorem}

\subsection{Lower Bound}
We have shown in the previous section that for any $n$-node graph we can construct
a witness cycle of length at most $4\frac{1}{3}n-4$.
In this section we complement this result with the lower bound $2.8n$:

\begin{theorem} \label{th:lb}
For any non-negative integers $n$, $k$ and $l$ such that, $n=5k+l$ and $l<5,$
there exists an $n$-node graph for which any witness cycle
is of length $14k+2l$.
\end{theorem}


\section{Periodic traversal with constant memory}

In this section we focus on the construction of a tour in arbitrary
undirected graphs to be traversed by an agent equipped with
a constant memory.
%
%
%
The main idea of the periodic graph traversal mechanism
proposed in \cite{Ilc06}, and further developed in \cite{GKMNZ08},
is to visit all nodes in the graph while traversing along an
{\sl Euler tour} of a (particularly chosen) spanning tree.
In \cite{Ilc06} the agent after entering a node $v$ in the tree via port $1,$
which always leads to the parent in the spanning tree,
visits recursively all subtrees accessible from $v$ via ports
$2,\dots,i+1,$ where $i$ is the number of children of $v.$
When the agent returns from the last ($i$th) child it either:
(1) returns to its parent via port 1, when $i+1$ is also the degree of
$v$ (i.e., $v$ is saturated in $T$); or
(2) it attempts to visit another child of $v$ adopting the edge $e$
associated with port $i+2.$
In case (2) the agent learns at the other end of $e$ that the port number is
different from $1$, i.e.,
the agent is not visiting a legal child of $v$.
The agent first returns to $v$ and
then immediately to
the parent of $v,$ where it continues the tree traversal process.
In these circumstances, the edge $e$ is called a {\sl penalty edge}
since $e$ does not belong to the spanning
tree and an extra cost has to be charged for accessing it.
Since the spanning tree
has $n-1$ edges, and at each node the agent can be forced to examine a
penalty
edge, the number of steps
performed by the agent (equal to the length of the periodic tour)
may be as large as $4n-2$ ($n-1$ edges of the spanning tree and $n$
penalty edges, where each edge is traversed in both directions).
The main result of \cite{GKMNZ08} is the efficient construction of a specific
spanning tree
supported by a more advanced visiting mechanism stored in the agent's memory.
They showed that the agent is able to avoid penalties
at a fraction of at least $\frac{1}{8}n$ nodes. This in turn gave the
length of
the periodic tour not larger than $3.75n$.

In what follows we show a new construction of the spanning tree, based
on the earlier three-layer partition. This, supported by a new labeling
mechanism together with slightly
increased memory of the agent, allows to avoid penalties at
$\frac{1}{4}n$ nodes resulting in a periodic tour of length
$\le 3.5n.$
In the new scheme some leaves in the spanning tree are connected with
their parents
via port 2 (in \cite{GKMNZ08} this port is always assumed to be 1). The
rationale behind
this modification is to treat edges towards certain leaves as penalty
edges (rather than
the regular tree edges) and in turn to avoid visits beyond these leaves,
i.e., to avoid
unnecessary examination of certain penalty edges.

Recall that the nodes of the input graph can be partitioned into three sets
$X,Y$ and $Z$ where
all nodes in $X$ and $Y$ are spanned by a backbone tree, see
section~\ref{s:3l-partition}.
The spanning tree $T$ is obtained from the backbone tree by connecting
every node in $Z$
to one of its neighbors in $Y.$ Recall also that every node $v\in X$ is
{\sl saturated}, i.e.,
all edges incident to $v$ in $G$ belong also to the spanning tree.
Every node in $Y$ that lies on a path in $T$ between two nodes in $X$
is called a {\sl bonding node}. The remaining nodes in $Y$
are called {\sl local}.

\noindent
{\bf Initial port labeling}
When the spanning tree $T$ is formed, we pick one of its leaves as the
root $r$ where the two ports located on the tree edge
incident to $r$
are set to 1. Initially, for any node $v$ the port leading to the parent
is set to 1 and ports
leading to the $i$ children of $v$ are set to $2,\dots,i+1,$ s.t., the subtree
of $v$ rooted in child $j$
is at least as large as the subtree rooted in child $j',$ for all $2\le
j<j'\le i+1.$
All other ports are set arbitrarily using distinct values from the range
$i+2,\dots,d_v,$ where $d_v$ is the degree of $v.$
Later, we modify allocation of ports at certain leaves of the spanning
tree located in $Z.$
In particular we change labels at all children having no other
leaf-siblings in $T$ of bonding nodes
(see, e.g., node $w_{1}$ in Figure~\ref{f:w1w2}), as well as in single
children of local nodes,
but only if the local node is the last child of a node in $X$ that has
children on its own (see, e.g., node $w_{2}$ in Figure~\ref{f:w1w2}).

\begin{figure}
\begin{center}
\includegraphics[width=5in]{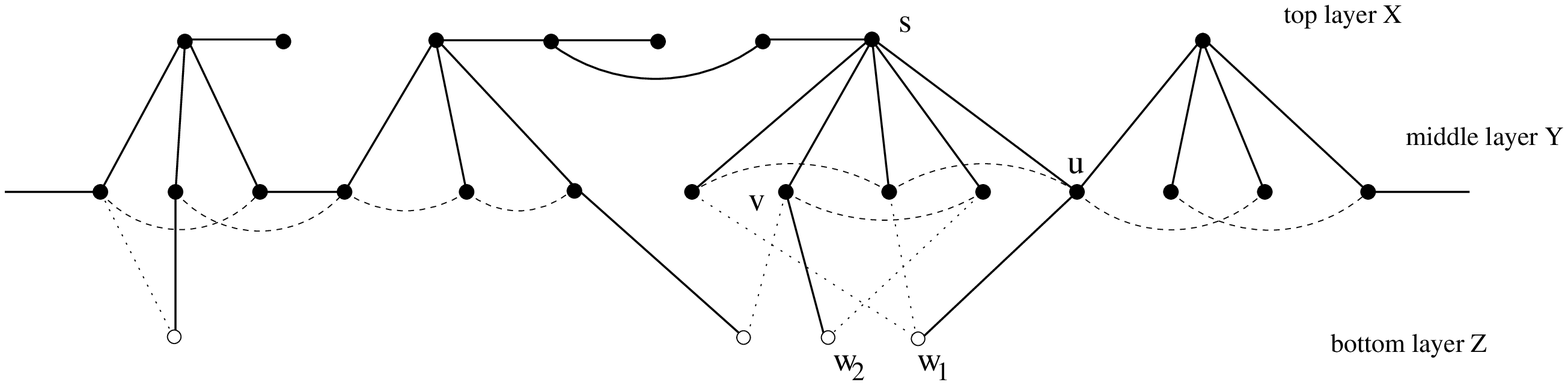}
\end{center}
\caption{Fragment of the spanning tree with the root
located to the right of $w_1$ and $w_2.$\label{f:w1w2}}
\end{figure}

\noindent
{\bf Port swap operation} Recall that every leaf $w$ located at the
level $Z$ has also an incident edge $e$ outside of $T$ that leads to
some node $v$ in $Y$
(property 2 of the three-layer partition).
When we swap port numbers at $w$, we set to 2 the port on the tree edge
leading to
the parent of $w.$ We call such edge a {\sl sham penalty edge} since it
now pretends
to be a penalty edge while, in fact, it connects $w$ to its parent in
the spanning tree $T$.
We also set to 1 the port number on the lower end of $e.$ All other port
numbers at $w$
(if there are more incident edges to $w$) are set arbitrarily.
After the port swap operation at $w$ is accomplished we also have to ensure
that the edge
$e$ will never be examined by the agent, otherwise it would be wrongly
interpreted as
a legal tree edge, where $v$ would be recognised as the parent of $w.$
In order to avoid this problem we also set ports at $v$ with greater
care. Note that $v$
has also an incident horizontal edge $e'$ outside of $T$
(property 1 of
the three-layer partition). Assume that the node $v$ has $i$ children in
$T.$ Thus if we
set to $i+2$ the port on $e'$ (recall that port 1 leads to the parent of $v$ and
ports $2,..,i+1$ lead
to its children) the port on $e$ will have value larger than $i+2$ and
$e$ will never be accessed by the agent.
Finally note that the agent may wake up in the node with a sham penalty edge
incident to it.
For this reason we introduce an extra state to the finite state
automaton ${\cal A}$
governing moves of the agent in~\cite{GKMNZ08} to form a new
automaton~${\cal A^+}$.
While being in the wake up state the agent moves across the edge
accessible via port 1 in
order to start regular performance (specified in~\cite{GKMNZ08}) in a node
that is
not incident to the lower end of a sham penalty edge.

\begin{lemma}\label{l:newlabel}
The new port labeling provides a mechanism to visit all nodes in the graph
in a periodic manner
by the agent equipped with a finite state automaton ${\cal A^+}$.
\end{lemma}

\begin{theorem}\label{th:3.5n}
For any undirected graph $G$ with $n$ nodes, it is possible to compute
a port labeling such that
an agent equipped with a finite state automaton ${\cal A^+}$ can visit
all nodes in $G$ in a periodic manner with a tour length that is no
longer than $3\frac{1}{2}n-2.$
\end{theorem}

Note that in the model with implicit labels, one port at each node has
to be distinguished in order to break symmetry in a periodic order of ports. This is
to take advantage of the extra memory provided to the agent.

\section{Further discussion}
\label{s:conclusion}
Further studies on trade-offs between the length of the periodic tour and
the memory of a mobile entity are needed.
The only known lower bound $2n-2$ holds independently of the size of the
available memory, and it refers to trees.
This still leaves a substantial gap in view of our new $3.5n$ upper bound.
Another alternative would be to look for as good as possible tour
for a given graph, for example, in a form of an approximate solution. Indeed,
for an arbitrary graph, finding the shortest tour may correspond to
discovering a Hamiltonian
cycle in the graph, which is NP-hard.

\paragraph{Acknowledgements.}
Many thanks go to Adrian Kosowski, Rastislav Kralovic, and Alfredo Navarra
for a number of valuable discussions on the main themes of this work.


\nocite{}
\bibliographystyle{splncs}
\bibliography{biblio-periodic-explo}



\section{Appendix}

\subsection{An example of the 3L-Partition}
The example from Figure~\ref{fig:3par-ex} illustrates the procedure
{\sc 3L-Partition} for a graph in Fig.~\ref{fig:3par-ex}(-). The graphs
from Fig.~\ref{fig:3par-ex}(a) through Fig.~\ref{fig:3par-ex}(e) present
the configuration after
each iteration of the loop, when a new node is chosen for saturation,
and the sets $X, Y, Z, P, R$ as well as the backbone tree $T_B$ are
modified accordingly. The saturated
nodes are $a, b, c, d$ and $e$, chosen from different sets $Y, Z$ and $P$.
In all configurations except (-) the content of each set $X, Y, Z, P, R$ is
represented at a different horizontal level.

\begin{figure}[htb] 
   \centering
   \includegraphics[scale=.90]{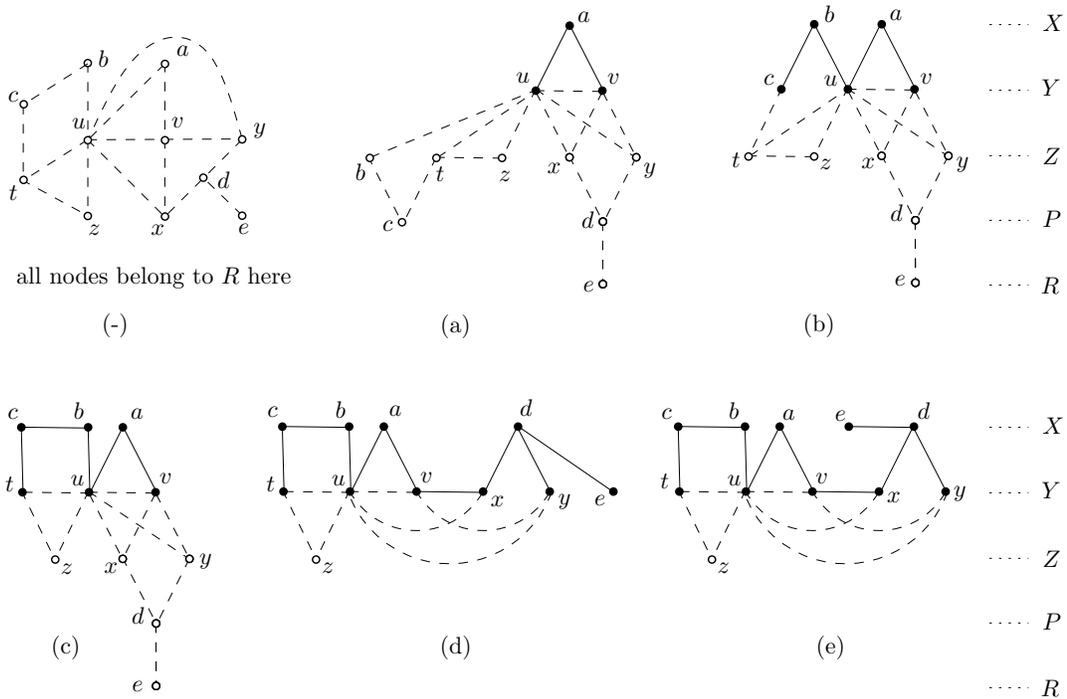}
   \caption{Example of functioning of procedure {\sc 3L-Partition}.
   Solid lines and black nodes belong to the backbone tree $T_B$.}
   \label{fig:3par-ex}
\end{figure}
\subsection{An example of local port ordering}
\begin{figure}[htb] 
   \centering
   \includegraphics[scale=0.80]{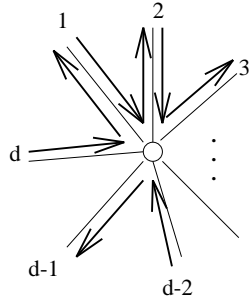}
   \caption{Ordering two bidirectional, two incoming and two outgoing underlying edges.}
   \label{Fig:first}
\end{figure}


\subsection{An example of a port numbering which induces a union of RH-cycles}
One can verify, that
if we exchange any two labels of one of the degree-three
nodes of this graph, the three cycles would merge to a single witness cycle.

\begin{figure}[h]
\includegraphics{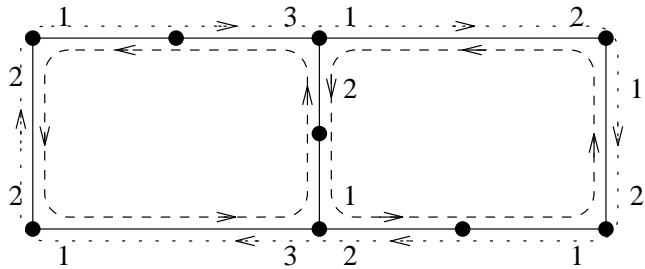}
\centering
\caption{Each node is RH-traversable, but the witness
ordering of nodes define several cycles, i.e., no cycle spans the whole graph.}
\label{Fig:2cycl}
\end{figure}

\subsection{Proof of Lemma~\ref{lem:algo}}
We show that the procedure creates a three-layer partition with a
distinguished backbone tree $T_{B}$.
Recall that the contents of sets $Y,Z,P$ and $R$ strictly depend on the content of $X.$
New nodes are inserted to $X$ gradually, one per {\sl round}, where each round
corresponds to a single execution of the main loop. We use the following
invariant. At the start of each round nodes in sets $X$ and $Y$ are spanned by a partial
backbone tree $T_B$ and a newly selected node for saturation is provided as $v$. At the end
of the first round the invariant is satisfied since $X$ contains only one node whose
neighbors in $G$ form $Y$ (step 3b) and all edges incident to it belong to $T_{B}$ (step 3c).
Assume now that
the invariant is satisfied at the beginning of some further round $i.$ When the newly
selected node $v$ is inserted to $X$ (step 3a) the content of other sets is recomputed
(step 3b). Note that $v$ is always selected, s.t., adding all edges incident to $v$ will
not form a cycle with edges in $T_{B}$.
If $v$ was chosen from $Y$ (this happens only
when $v$ has no horizontal incident edges), $v$ is already connected to $T_{B}$ so all
incident edges to $v$ (added in step 3c) will be connected to the rest of $T_{B}$ too.
Alternatively, if $v$ comes from $Z$ (this happens when all nodes in $Y$ have
horizontal edges outside of $T_{B}$) and $v$ has exactly one neighbor $w\in Y$
as soon as all edges incident to $v$ are inserted,
the new part of $T_{B}$ gets connected to the old one via the node $w.$
Finally, if $v$ was selected in $P$ (this happens when all nodes in $Y$ have
horizontal edges outside of $T_{B}$ and all nodes in $Z$ have at
least two neighbors in $Y$) then all edges incident to $v$ are inserted to $T_{B}.$
Note that when $v$ was moved to $X$ all its neighbors in $Z$ were
moved to $Y$ forming at least one new horizontal edge in $Y$ (formerly this edge
laid across sets $Y$ and $Z$). We use this new horizontal edge to connect a newly
formed star with the remaining part of $T_{B}.$
The procedure stops when it attempts to select a new node for saturation from an empty
set $P$ meaning that all nodes from $V$ are already distributed among $X,Y,Z$ for which,
according to our invariant, the backbone tree $T_{B}$ is already completed. \qed

\subsection{Proof of Theorem~\ref{lm:3layer-complexity}}
Each edge $\{v,w\} \in G$ is taken into consideration twice by the procedure, once
as the directed edge $(v,w)$ and the other time as the directed edge $(w,v)$. We prove
that for each directed edge the procedure performs a constant time task.

We suppose that each node of the graph is colored {\em red} when it has been already
tested for saturation (whether or not it was eventually included in set $X$) or
{\em green} otherwise. All nodes are initially colored green and put in set $R$.
During the execution of step (e) of \textsc{3L-Partition} procedure a green node
$v$ is selected for saturation from a set $Y, Z$ or $P$ (in that order). Depending
on the set to which belongs $v$ and the sets to which belong all its neighbors $v$
passes or fails the saturation tests. In particular:
\begin{enumerate}
 \item If $v \in Y$ then $v$ becomes saturated (and promoted to $X$) if none of
 its neighbors belongs to $Y$.
 \item If $v \in Z$ then $v$ becomes saturated if only one of its neighbors
 belongs to $Y$.
 \item If $v \in P$ then $v$ always becomes saturated.
\end{enumerate}
Moreover, each neighbor of $v$, whether it is green or red, may be promoted to a
higher ranking set among $Y, Z, P$ and $R$ depending on the result of the saturation
of $v$. This needs a second scan of the list of neighbours,
once the set into which $v$ is put has been determined. Hence $v$ is turned from
green to red as a result of a $O(d_v)$-time step. The \textsc{3L-Partition} procedure
terminates when all nodes are red so its overall complexity is $O(|E|)$. \hfill\qed

\subsection{Proof of Lemma~\ref{lem:only-RH}}

  Consider an arbitrary algorithm $\cA$ enabling an oblivious agent to
  periodically explore all trees. Let $f$ be its transition
  function. Fix an arbitrary $d>1$ and let $f_d$ be the function
  $i\mapsto f(s,i,d)$ from $\{1,\dots,d\}$ to $\{1,\dots,d\}$, where
  $s$ is the single state of the oblivious agent. Consider the
  $d+1$-node star of degree $d$. For $1\leq i \leq d$, let $v_i$ be
  the leaf reachable from the central node $u$ by the edge with port
  number $i$.

  For the purpose of contradiction, assume first that $f_d$ is not
  surjective. Let $i$ be a port number without pre-image. If the agent
  is started by the adversary in node $v_j$, with $j\neq i$, then the
  node $v_i$ is never explored. Therefore $f_d$ is surjective, and
  thus a permutation of the set $\{1,\dots,d\}$. Again for the purpose
  of contradiction, assume that $f_d$ can be decomposed into more than
  one cycle. Let $i$ be a port number outside $1$'s orbit (i.e. $1$
  and $i$ are not in the same cycle of the permutation). If the agent
  is started by the adversary in node $v_1$, then the node $v_i$ is
  never explored. Hence $f_d$ is a cyclic permutation, i.e., it is
  constructed with a single cycle. Since the equivalence classes of
  permutations (usually called conjugacy classes) correspond exactly
  to the cycle structures of permutations, the agent algorithm $\cA$
  is equivalent to the Right-Hand-on-the-Wall algorithm.

\subsection{Proof of Corollary~\ref{cor:TerseCycles}}
Note that after the execution of procedure \textsc{TerseCycles}   each node
of $v \in \overrightarrow{K}$ has an even number (different from zero)
of single edges incident to it. One can provide direction to all single
edges and port numbering at each node $v$, s.t., all the edges
outgoing from and incoming to $v$ belong to the same cycle. This is
done in two steps. First, the initial port numbering and
the direction of single edges are obtained via greedy selection of single
edges to form cycles. Later, if there is a node $v$ that belongs to two cycles
(based on single edges), the cycles are merged at $v$ via direct port number
manipulation. When this stage is accomplished, the set of nodes in $\overrightarrow{K}$
is partitioned into components, with all nodes in the same component belonging
to the same cycle based on single edges.
Note also that each component is at distance one from some other component,
where the components are connected by at least one two-way edge (this is a
consequence of the fact that each node has at least two single edges
incident to it). The two-way edge is used to connect the components.
Connecting successively pairs of components at distance one we end up
with a single component, i.e., a witness cycle spanning all the nodes.
Note that for each single edge introduced to $\overrightarrow{K}$ a two-way edge
from the spanning tree is reduced to single during the restore parity process.
This happens because single edges form a collection of stars and
at least one endpoint of each single edge (in a star) is free.
Thus the number of all edges in the witness cycle is bounded by $2n.$
\hfill\qed

\subsection{Proof of Theorem~\ref{th:4.3}}
Since $\overrightarrow{K}$ from line~8 contains $T$, it is a connected
spanning subgraph of $G$.
For each such component we apply procedure \textsc{3L-Partition},
obtaining three sets $X_{i},Y_{i}$ and $Z_{i}$, and a backbone tree
structure $T_{i}$. By lemma~\ref{lm:3-conditions} we can add single edges
incident to the nodes in $Y_{i}$ and pairs of single edges incident
to the nodes in $Z_{i}$ and then apply procedure \textsc{RestoreParity}
to each component $G_i$. Note that,
when each star $S_i$ is constructed, we may do it in such a way that no
path of length three or more is created. Indeed, otherwise we could remove
a middle edge of any path of length three and the set of spanned
nodes would remain the same. Hence $S_i$ is a forest of stars. Moreover we
can assume that only centers of such stars can be incident to edges forming $P_i,$
otherwise any edge leading to a leaf node incident to $P_i$ can be removed.
Consequently, after termination of the ``for'' loop, each
node of $G$ is RH-traversable in $\overrightarrow{K}$. Moreover,
since $\overrightarrow{K} \supseteq T$, $\overrightarrow{K}$ is two-way
connected, so the conditions of lemma~\ref{lm:main} are satisfied.
Hence, at the end of the algorithm $\mathfrak{C}$ contains a witness cycle.
\noindent

In order to bound the size of the witness cycle we will bound the number
of edges in $\overrightarrow{K}$. Note first that $2n-2$ edges are used in $T$
(i.e. $n-1$ two-way edges). Suppose first that for each component $G_i$, containing $n_i$
nodes of $G \setminus T$, no single edges were added in lines 4 and 5,
that is $P_i = \emptyset$ and $S_i = \emptyset$.
Hence, the call of procedure {\sc RestoreParity}  from line 6 did not
modify $G_i$. In consequence, $2(n_i-1)$ edges were added for $G_i$
or $2(n_1+n_2+ \dots +n_k)-2k$ in total. This value is maximized
for $k=1,$ giving $2n-2$ edges added in the ``for'' loop, and $4n-4$ total
edges in $\overrightarrow{K}$. The count remains the same if some $P_i \neq \emptyset$,
since exactly two edges were added for each node of $Z_i$ in line 4.
\noindent

Suppose now that $S_i \neq \emptyset$, in line 5, for some components $G_i$.
For each endpoint  $v \in Y_i$ of a star belonging to $S_i$ and a single edge $e$
added for $v$ in $S_i$ in line 5, we check whether there is some other edge
that was reduced (from two-way to single) during the call of
procedure {\sc RestoreParity} in line 6. This happens when $v$ is not
incident to a horizontal edge of the backbone tree $T_i$, since one of
the edges incident to $v$ will then become single. Thus the addition
of $e$ is done at no extra cost, i.e., the total number of edges remains the same.
However, when
two endpoints of a horizontal edge are incident to two edges of $T_i$, only
one such edge will be amortised. Consider then a collection of single horizontal edges,
belonging to the backbone tree $T_i$ with edges of $S_i$ incident to both of
their end-points.
The collection forms a forest. In each tree pick a root arbitrarily and
repeat the following
process until there is only one edge left in it. Take an arbitrary leaf and amortise
the edge of $S_i$ incident to it with the tree edge leading to the parent of the leaf.
Remove the leaf and the edge that leads to its parent from further consideration.
Note that in this case amortisation is one to one. When this process is accomplished
each tree is reduced to one edge. In other words we have a collection of independent
single horizontal edges belonging to the backbone tree. Note that each such edge
is associated with two independent edges of $S_i$. Clearly the worst case
happens when the forest was formed by independent single edges. This implies that the
number of such horizontal edges is not larger than $\frac{n_{i}}{3}$.

Taking into consideration the maximal penalty that we have to pay for edges
added in line 5 of the algorithm, the number of edges forming $\overrightarrow{K}$ is
bounded by $4 \frac{1}{3}n-4$. \hfill\qed

\subsection{A lower bound example for the {\sc FindWitnessCycle} algorithm}
The example from Figure~\ref{fig:parachute} shows that the bound from
Theorem \ref{th:4.3} is tight (up to an additive constant) for our algorithm.
More precisely, the image shows that there exist graphs on which our algorithm
may produce a witness cycle of size $4\frac{1}{3}n - 7$.
\begin{figure}[!htb]
  \begin{center}
\includegraphics[width=11cm]{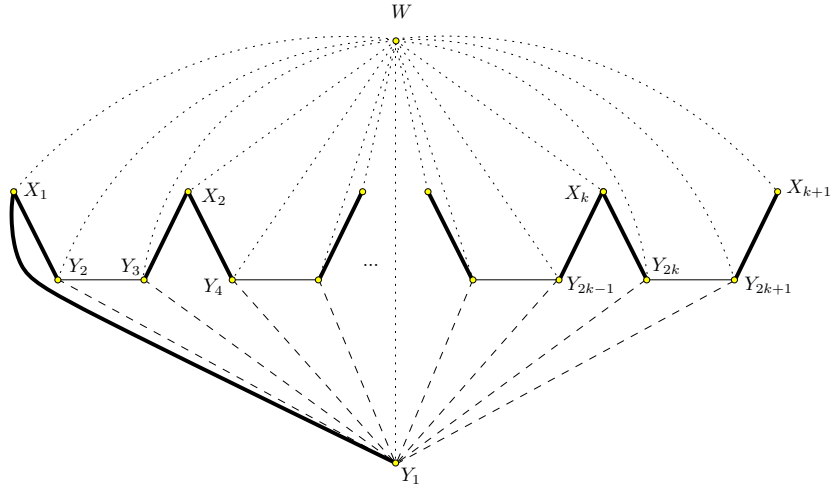}
  \end{center}
  \caption{Example of a graph for which our algorithm
           gives a witness cycle of size not smaller than
           $4\frac{1}{3}n-7$}
  \label{fig:parachute}
\end{figure}
The main part of the graph containing $n=(3k+1)$ nodes consists of $k$ copies
of four nodes
$X_{i}Y_{2i}Y_{2i+1}X_{i+1}$, for $i=1,2,\dots, k$, where the last node of
each but the last copy is identified with the first node of the next copy
(cf. Fig.~\ref{fig:parachute}). Moreover, an extra node $Y_1$ is
adjacent to each of the nodes $Y_{2},Y_{3},\dots, Y_{2k+1}$,
and a node $W$ is adjacent to all other nodes in the graph.
Suppose that the star at node $W$ is chosen by the algorithm
as the spanning tree $T$, represented by the dotted edges in the
picture. The procedure {\sc 3L-Partition}  locates nodes
$X_{1},X_{2},\dots, X_{k}$ in set $X$ and
the nodes $Y_{1},Y_{2},\dots, Y_{2k+1}$
in set $Y$ (set $Z$ is empty). Suppose that the spanning tree is the
path $Y_{1}X_{1}\dots X_{k+1}$ - represented by the solid edges in
Fig.~\ref{fig:parachute}. Since the algorithm adds one horizontal
edge for each node from class $Y$, all edges incident to $Y_1$
are added to the structures. It is easy to see, that the parity
restoring procedure will chose the edges $Y_iY_{i+1}$ as the single
edges of the structure. In consequence, only $2k$ dashed edges and $k$ thin
solid edges in Fig.~\ref{fig:parachute} are chosen as single
edges --- all other edges (i.e. $3k+2$ dotted edges and
$2k+1$ bold solid edges) are taken as two-way edges. This results in
a witness cycle of size $13k+6$, i.e.\ containing
$4\frac{1}{3}n-7$ edges.

\begin{figure}[htbp] 
   \centering
   \includegraphics{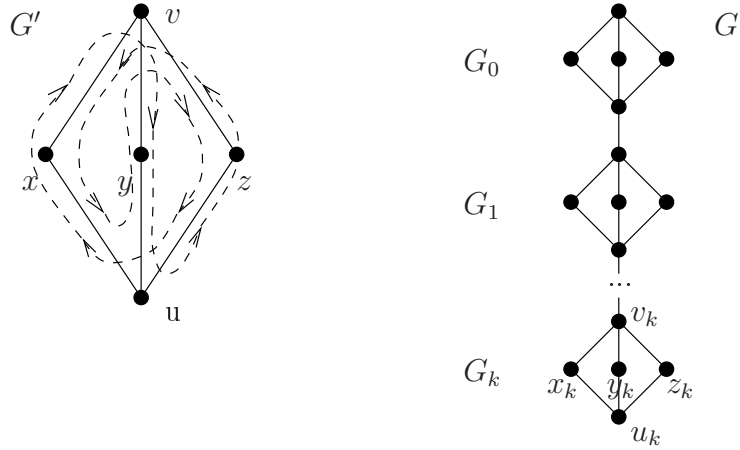}
   \caption{The lower bound based on diamond graphs.}
   \label{fig:lb}
\end{figure}

\subsection{Proof of Theorem~\ref{th:complexity}}
In $O(|E|)$ time we can find a spanning tree $T$ of $G$ and the connected
components of $G \setminus T$. By lemma \ref{lm:3layer-complexity}, for each
connected component $G_i$ having $n_i$ nodes and $e_i$ edges, the call of
the procedure {\sc 3L-Partition}  terminates in $O(e_i)$ time. The construction
of sets $P_i$ in line 4 and set $S_i$ in line 5 as well as the call of procedure
{\sc RestoreParity} in line 6 are completed in $O(n_i)$ time.
Altogether, the ``for'' loop terminates in $O(|E|)$ time. The construction
of $\overrightarrow{K}$ in line 8 and $\mathfrak{C}$ in line 9 are done
in time proportional to their sizes, i.e., $O(n)$.

We show now that line 10, where the rules {\em Merge3} and {\em EatSmall} are
repeatedly applied, may be performed within $O(n)$ time. We chose any ordering
$\gamma$ of cycles and we attach to each edge a label corresponding to the cycle
to which the edge belongs. Let $C^*$ be the largest cycle according to $\gamma$
and $v$ be any vertex of $C^*$. We perform repeatedly rules {\em Merge3}
(resulting cycle obtaining rank of $\gamma(C^*)$) and  {\em EatSmall} to
vertex $v$ until no longer possible. Each time we traverse the edges of
the cycle (or a part of the cycle) added to $C^*$ and change their labels to $\gamma(C^*)$.
When neither {\em Merge3} nor {\em EatSmall} is applicable to $v$ we proceed to vertex $v'$
- the actual successor of $v$ in $C^*$ - and repeat the procedure of applying rules
{\em Merge3} or {\em EatSmall} to $v'$. Although $C^*$ changes dynamically and
some vertices may be traversed many times we end up by traversing all vertices
eventually in $C^*$. By lemma~\ref{lm:main} at the end of this process
$C^*$ becomes a witness cycle.
Note that the complexity of each {\em Merge3} and {\em EatSmall} operation is
proportional to the number of edges added to $C^*$.
By theorem~\ref{th:4.3} the overall complexity of line 10 is $O(n)$. \hfill\qed

\subsection{Proof of Theorem~\ref{th:lb}}
Consider first a single {\sl diamond graph} $G',$ see left part of
Figure~\ref{fig:lb}. W.l.o.g., we can assume that we start the traversal
through $(v, x)$. Consider the successor of $(x,u)$. Also, w.l.o.g., we
can take $(u,y)$ as the successor. Now there is only one feasible
successor of $(y,v)$ and that is $(v, z).$ All other edges violate
either RH-traversability ($(v,y)$) or leave $z$ unvisited.
Similarly, the only possible successor of $(z,u)$ is $(u,x)$
($(u,y)$ has already been traversed with a different predecessor,
and $(u,z)$ violates RH-traversability), of $(x,v)$ is $(v,y)$ and
of $(y,u)$ is $(u,z)$. Therefore, each edge of $G'$ must be used in
both directions.

Consider now a chain of diamond graphs from the right side
of Figure \ref{fig:lb}, starting the graph traversal at node $v_0$.
From the fact that each edge in the witness cycle is traversed at
most twice (one time in each direction) it follows that when returning
from $v_i$ to $u_{i-1}$, all nodes in $G_i$ (as well as in all $G_j,$
for $j>i$) must have been visited.
Note that from RH-traversability it follows that the successor of
$(u_{i-1}, v_i)$ cannot be the same (in reverse direction) as the
predecessor of $(v_i, u_{i-1})$, and similarly the successor of
$(v_{i}, u_{i-1})$ cannot be the same as the predecessor of $(u_{i-1},v_{i})$.
In turn this means that the analogous arguments (as used in $G'$)
apply also to each $G_i$, therefore all edges of $G$ must be
traversed in both directions.

The theorem now follows directly for $n=5k$. If $n$ is not a multiple of
$5$, an extra path of $l$ nodes can be added to $u_k$ to satisfy the
claim of the theorem. \hfill\qed

\subsection{Proof of Theorem~\ref{l:newlabel}}
%
Note that it is enough to prove that no difficulty arises at nodes with
numbers affected by the modified labeling scheme.

\noindent
{\bf Case C1:}
Consider first the case when the port numbers are swapped at some node
$w_1$ which is
a single child of a bonding node $u$ (see Fig~\ref{f:w1w2}).
When during traversal the agent returns from the subtree rooted in a
child of $u$ accessible
via port $i-1$ it enters via port $i$ the edge leading to $w_1.$ This edge
is interpreted as
a penalty edge and the agent after visiting $w_1$ returns immediately to
$u$ and then it goes
with no further action to the parent of $u.$ Note that if the labeling
was not changed the agent
would act similarly, however it would examine additionally a penalty
edge located at $w_1.$
Thus thanks to the new labeling scheme we save one penalty at the node $w_1.$

\noindent
{\bf Case C2:}
Consider now the second case when the port numbers are swapped at a
single child $w_2$
of a local node $v$, s.t., $v$ has no siblings different from leaves
to its right (accessible via larger ports), see Fig~\ref{f:w1w2}.
Assume that $s$ is the (saturated) parent of $v$ and port
$i$ at $v$ leads to $w_2.$
When during traversal the agent returns from the subtree
rooted in a child of $v$ accessible via port $i-1$ it enters via port
$i$ the edge leading
to $w_2.$
When it learns that
the port label at $w_2$ is different from $1$ it interprets the sham
penalty edge linking
$v$ and $w_2$ as the penalty edge. The agent returns immediately to $v$
while switching to the leaf recognition state ($v$ would be interpreted as
the first leaf of $s$).
This means that all remaining leaves accessible from $s$ (if any) will
be visited at no extra charge, i.e., without paying penalty at them.
Thus the agent does not miss the node $w_2$ and it also saves penalty
at $w_2$ and possibly at all leaves that are siblings of $v.$ \hfill\qed

\subsection{\bf Proof of Theorem~\ref{th:3.5n}}
The main line of the proof explores the fact that the fraction of nodes
at which the agent
manages to save on penalties is at least $\frac{1}{4}.$ The proof is
split into global
and local amortisation arguments.

\vspace*{0.2cm}
\noindent
{\bf Global amortisation} [saturated nodes amortise all bonding nodes
and single children
of saturated nodes]
%
%

Note that in a three-layer partition with $k$ saturated nodes there are
at most $2k-2$ bonding nodes, since introduction of a new saturated node
implies creation of at most two bonding nodes.
Note also that there are at most $k$ single leaves (with no siblings)
that are children of
saturated nodes. In the global amortisation argument we assume that at
these nodes, i.e., all bonding nodes and all single leaves of
saturated nodes, in the worst case the agent always pays
penalty (examines the penalty edge).
Fortunately, all of these $\le 3k-2$ nodes ($2k-2$ bonding nodes and $k$
single leaves
of saturated nodes) can be amortised by $k$ saturated nodes. Thus, as
required,
the fraction of nodes where the agent does not pay penalty is
$\frac{1}{4}.$
For all other nodes in $T$ we use the local amortisation argument.

\begin{figure}
\begin{center}
\includegraphics[scale=0.60]{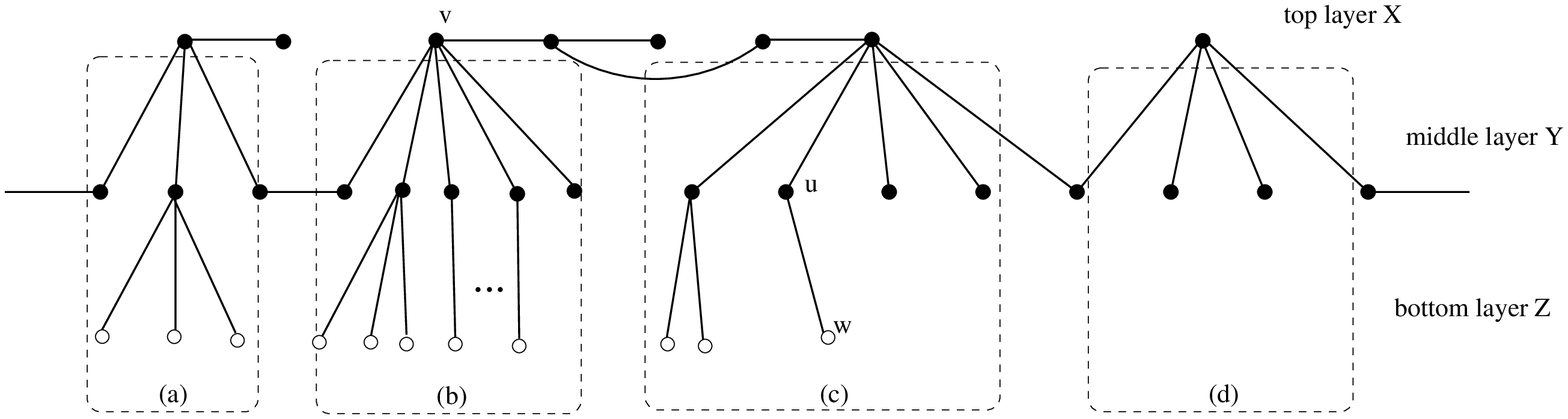}
\end{center}
\caption{Local amortisation argument -- cases (a), (b), (c) and (d).\label{f:abcd}}
\end{figure}

\vspace*{0.2cm}
\noindent
{\bf Local amortisation} [direct amortisation of nodes within small
subtrees]

\noindent
The local amortisation argument is used solely on two-layer subtrees
accessible
from saturated nodes, i.e., formed of local nodes and (possibly) their
children,
cases (a), (b), (c), and (d), see Figure~\ref{f:abcd},
as well as on leaves accessible from bonding nodes, cases (e) and (f).

The local amortisation argument involving local nodes is split into cases
(a), (b), (c), and (d) in relation to the size of subtrees rooted in
local nodes.
We start the analysis with the largest subtrees in case (a) and gradually
move towards smaller structures in cases (b) and (c), finishing with single
local nodes in case (d).

\noindent
{\bf (a)} Consider any subtree $T_S$ with at least two children rooted in
a local node. In this case the initial labeling remains unchanged.
During traversal of $T_S$ the agent pays penalties at the local node and
at its
first child where it switches to the leaf recognition state. In this state
no further penalties at the leaves of $T_S$ are paid.
Since the number of children $i\ge 2$ the fraction of nodes in the subtree
without penalties is at least $\frac{1}{3}.$

\noindent
{\bf (b)} Consider now the case where a saturated node $v$ has
at least two children (local nodes) with single children (two {\sl extended
leaves} according to the notation from~\cite{GKMNZ08})
accessible from $v$. In this case the number of penalties
paid during
traversal of all extended leaves is limited to two since the penalties
are paid at
both nodes of the first extended leaf where the agent switches to
extended leaves
recognition state. The remaining nodes of the extended leaves are visited
at no extra cost.
In this case the fraction of nodes without penalties is at least
$\frac{1}{2}.$

\noindent
{\bf (c)} Consider now the case where a saturated node has only one
extended
leaf (a local node $u$ and its single child $w$) possibly followed by
some regular
leaves formed of local nodes. In this case the initial labeling is
changed and
the sham penalty edge $(u,w)$ is introduced (case C1 in the proof of
lemma~\ref{l:newlabel}).
When the agent visits the extended leaf it enters the
sham penalty edge
interpreting it as the penalty edge. Thus the penalty is paid only at
the local
node $u.$ Moreover if $u$ has sibling leaves all of them are visited at
no extra
cost since after visiting a sham penalty edge the agent is in the leaf search
state (\cite{GKMNZ08}).
Thus also in this case the fraction of nodes where the penalty is not paid
is at least $\frac{1}{4}.$

\noindent
{\bf (d)} It may happen that a saturated node has several
children that
are leaves in $T$ not preceded by an extended leaf. In this case the
penalty is paid
only at the first leaf and all other leaves are visited (in leaf search
mode)
at no extra cost.
(Recall that the case when a saturated node has only one child that is a
leaf in $T$
was already considered as a part of the global amortisation argument.)

The remaining cases of the local amortisation argument refers to the leaves
accessible via bonding nodes.

\noindent
{\bf (e)} When a bonding node has at least two children (all children
are leaves)
during traversal the agent pays penalty only at the first child while
all other children
are visited at no extra cost (thanks to the leaf search state).
Thus the fraction of nodes (leaves) where the penalty is avoided is at
least $\frac{1}{2}.$

\noindent
{\bf (f)} Finally consider the case where a bonding node $u$ has exactly
one child $w$
(case C2 from the proof of lemma~\ref{l:newlabel}).
In this case thanks to the sham penalty edge $(u,w)$ no penalty is paid at $w,$
i.e., the fraction of nodes without penalties is $1.$

In conclusion, the fraction of nodes at which the penalty is avoided is
bounded
from below by $\frac{1}{4}$ in all considered cases. Thus the number of
visited penalty
edges is bounded by $\frac{3}{4}n.$ Since the number of edges in the
spanning tree is
$n-1$ the agent visits at most $1\frac{3}{4}n$ edges where each edge is
visited in
both directions. This concludes the proof that the length of the tour is
bounded by $3\frac{1}{2}n-2.$ \hfill\qed


\end{document}